\documentclass{article}

\usepackage{arxiv}

\usepackage[utf8]{inputenc} 
\usepackage[T1]{fontenc}    
\usepackage{hyperref}       
\usepackage{url}            
\usepackage{booktabs}       
\usepackage{amsfonts}       
\usepackage{nicefrac}       
\usepackage{microtype}      
\usepackage{lipsum}
\usepackage{enumitem}
\usepackage{amsthm,amsmath}
\usepackage{nccmath}
\newtheorem{theorem}{Theorem}

\newtheorem{lemma}[]{Lemma}

\usepackage{comment}
\PassOptionsToPackage{square, sort, numbers}{natbib}
\usepackage{times,float}
\usepackage{algorithm}
\usepackage[noend]{algpseudocode}
\usepackage{multirow}
\usepackage{subfigure}
\usepackage{amssymb}
\usepackage{xspace}
\usepackage{todonotes}

\newcommand{\mc}{\mathcal}
\newcommand{\mbb}{\mathbb}

\newcommand{\mbf}{\mathbf}

\newcommand{\ALOI}{\texttt{ALOI}\xspace}
\newcommand{\News}{\texttt{20NG}\xspace}
\newcommand{\CT}{\texttt{Covertype}\xspace}
\newcommand{\LComments}{\texttt{Comments}\xspace}

\newcommand{\Test}{Test\xspace}
\newcommand{\HTNN}{\texttt{HTNN}\xspace}
\newcommand{\LT}{\texttt{LT}\xspace}
\newcommand{\ANN}{\texttt{ANN}\xspace}
\newcommand{\PERCH}{\texttt{PERCH}\xspace}
\newcommand{\Clear}{\texttt{Clear}\xspace}
\newcommand{\LQ}{\texttt{LQ}\xspace}
\newcommand{\Spam}{\texttt{Spam}\xspace}
\newcommand{\EV}{\texttt{EV}\xspace}
\newcommand{\ASPD}{\texttt{AEV}\xspace}
\newcommand{\AEV}{\texttt{AEV}\xspace}
\newcommand{\RP}{\texttt{RP}\xspace}
\newcommand{\LSH}{\texttt{LSH}\xspace}
\newcommand{\SVM}{\texttt{SVM}\xspace}
\def\Re{\mathbb{R}}

\title{Efficient Hierarchical Clustering for Classification and Anomaly Detection}

\author{
  Ishita Doshi\\
  LinkedIn\thanks{This work was done while Ishita Doshi was studying at IIT Gandhinagar, and during her Internship at LinkedIn. She is currently working at LinkedIn.}\\
  \texttt{idoshi@linkedin.com} \\
   \And
   Sreekalyan Sajjalla\\
   Flipkart\thanks{This work was done while Sreekalyan Sajjalla was at LinkedIn. He is currently working at Flipkart.}\\
   \texttt{kalyan1337@gmail.com}\\
   \And Jayesh Choudhari\\
   University of Warwick\\
   \texttt{jayesh.choudhari@warwick.ac.uk} \\
   \And Rushi Bhatt\\
   LinkedIn\\
   \texttt{rbhatt@linkedin.com}\\
   \And Anirban Dasgupta\\
   IIT Gandhinagar\\
   \texttt{anirbandg@iitgn.ac.in}\\
}

\begin{document}
\maketitle

\begin{abstract}
We address the problem of large scale real time classification of content posted on social networks, along with the need to rapidly identify novel spam types. Obtaining manual labels for user generated content using editorial labeling and taxonomy development lags compared to the rate at which new content type needs to be classified. We propose a class of hierarchical clustering algorithms that can be used both for efficient and scalable real-time multiclass classification as well as in detecting new anomalies in user generated content. Our methods have low query time, linear space usage, and come with theoretical guarantees with respect to a specific hierarchical clustering cost function~\cite{Dasgupta:2016} (Dasgupta, 2016). We compare our solutions against a range of classification techniques and demonstrate excellent empirical performance.
\end{abstract}

\section{Introduction}

{For any social networking site, understanding the trends among user generated content is of central importance. Real time identification and classification of new user generated content can help in various tasks such as recommending interesting content to users, ranking the generated content on other users news feed, and most importantly, in the identification of any spam attacks on the site.} How do we detect spam within user generated content with low latency, on a large scale and while also detecting novel attacks and their trends? This is the motivating question for this work.

There are a number of non-trivial aspects to this problem. In popular social networks, the number of comments posted per day orders up to millions. With this kind of a growth, both in the volume of content as well in the variety of topics, it is infeasible to obtain labels at scale since editorial labeling is a time consuming and expensive process. This gives rise to the need of combining labels generated through various sources--- labels obtained through user flagging, labels generated through manual inspection by editors and labels generated through machine classification. This however, remains a noisy signal at best. The severe scarcity of on-time, quality, levels at a fine-grained scale leads content management teams to do classification at multiple levels of granularity as demonstrated in Figure~\ref{fig:scamb}. This allows one to view these classes as a hierarchy.

\begin{figure}
\centering
\subfigure[]{\includegraphics[width=0.22\textwidth]{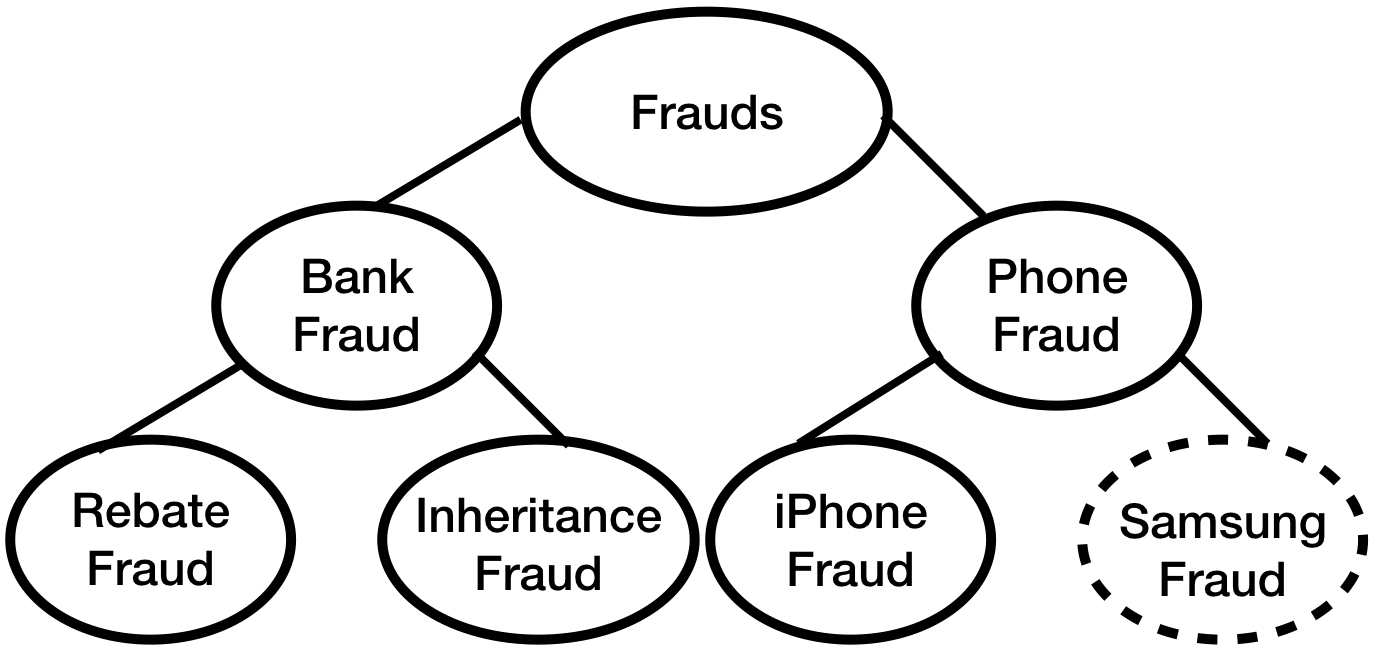}\label{fig:scamb}}
\subfigure[]{\includegraphics[width=0.22\textwidth]{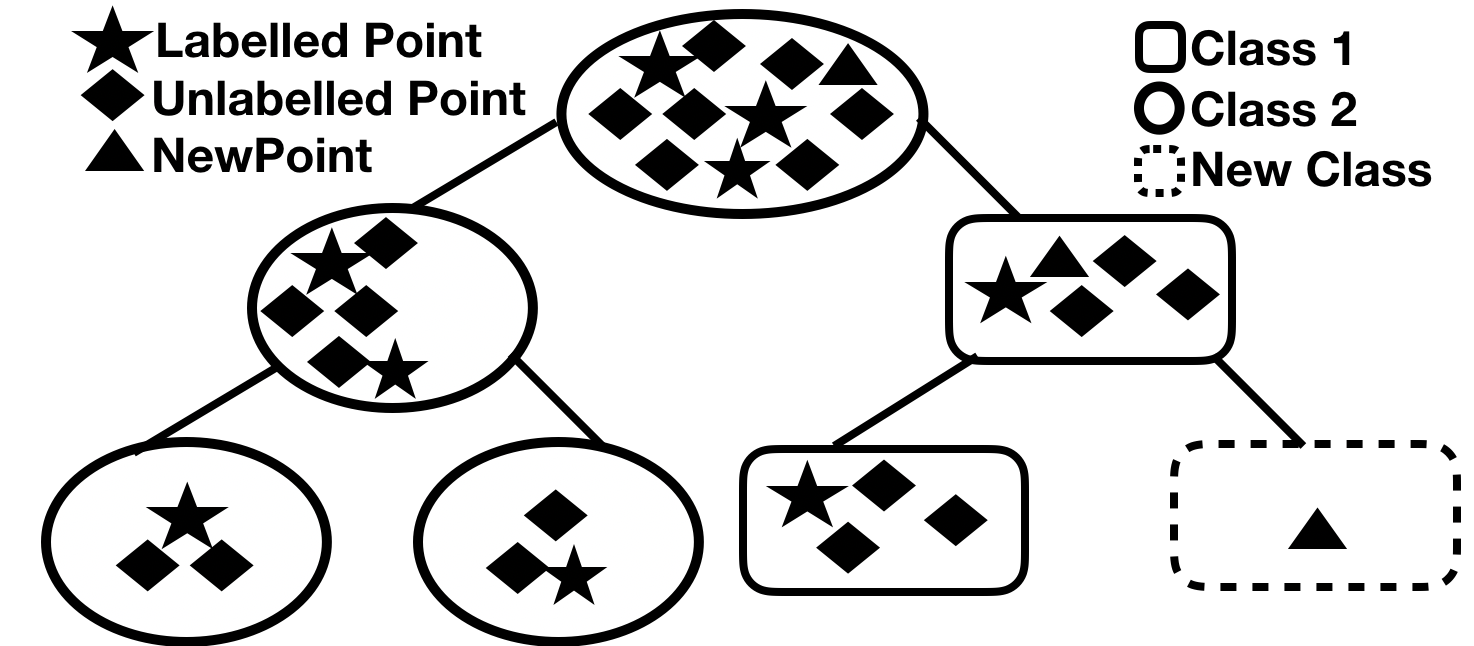}\label{fig:scama}}
\caption{(a) Hierarchical clustering helps determine necessary action on emerging spam ("Samsung Phone fraud") by using the parent category, and facilitates editorial decision making. (b) Labels in parent classes allow us to take necessary action while awaiting new a new class label and also allows us borrow strength when labeling is sparse}\label{fig:scam}
\end{figure}

Content hierarchy is critical to our application in two ways. First, we use it to help human editors determine whether a new class label is needed, and if so, where in the hierarchy: hierarchical clustering allows us to identify a new sub-cluster in the hierarchy which is then passed on to human editors with the taxonomy context to rapidly decide if a new class label is required. Figure~\ref{fig:scamb} shows an example scenario. Let us assume the site has so far only seen scams about selling iPhones (``iPhone Fraud" cluster in Figure~\ref{fig:scamb}). Let us also assume that a novel, and somewhat similar, spam attack for Samsung phones occurs. With an existing hierarchical clustering available, we can identify the new, potentially fast-growing, content cluster to be under the ``Phone Fraud" sub-tree and respond to the attack based on our existing policy on iPhone sale frauds. A new subcategory for ``Samsung Fraud" can be created through hierarchy clustering and exemplars. The new cluster along with hierarchical context can be rapidly sent for manual labeling. 

The second way in which a content hierarchy helps is in dealing with label sparsity in classes. About 9\% of all generated content is spam~\cite{spam}, yet user flagging is generally much sparser. 

Use of hierarchical clustering compensates for sparse labeling in new clusters by borrowing strength, or treatment, from parent nodes~\cite{AgarwalEtAl:2007}. Let us return to our example scenario in Figure~\ref{fig:scama}. We can estimate prevalence of spam in a node either editorially or via user flagging. If the prevalence of spam is high in a parent cluster, the new child cluster is likely to be spam. In this case, even before users start flagging the content or before editorial judgements are taken, appropriate action based on the policies for parent and sibling clusters can be taken.

Motivated by the problem of creating a content hierarchy with the above needs in mind, in this work, we show a principled approach to solving this problem. We show that the above setting, having specific requirements that necessitate combining large scale hierarchical clustering with high precision multi-class classification can be met in a principled manner using algorithmic techniques from hierarchical and spectral clustering. 
The techniques used are both scalable as well as equipped with theoretical guarantees. Using the resulting content hierarchy, of which only elected nodes may then be explicitly labelled by expert labelers, we then accomplish the dual task of classifying new content as well as determining whether a new class label is needed for newly arriving user generated content.

\subsection{Our contributions}

We start with the hierarchy cost defined by Dasgupta~\cite{Dasgupta:2016}. While this is an unsupervised method, our empirical results show that optimizing this cost function does indeed achieve a separation of the underlying classes (Figure~\ref{fig:zoo}) . However, getting the optimal hierarchy for this measure is NP-hard. We provide worst case bounds on the clustering cost of our scalable partitioning strategies as a function of the optimal hierarchical cost~\cite{Dasgupta:2016} in Section~\ref{sec:theory}.  
 
Our method first builds a hierarchical clustering using eigenvectors and approximate eigenvectors and next utilizes any available labeled instances in leaf nodes to carry out multiclass classification. The same clustering is also used to detect novel, fast growing, content types. 

Through extensive benchmarks and in-production measurements, we demonstrate that our method combining unsupervised hierarchical clustering with near-neighbor classification in leaf nodes best suits our application needs described above. Our method scales well at build and run time, it is amenable to adapting to new trends through human-in-loop, while allowing for fallback to a parent category for novel content types. 

One may argue that flat clustering can also allow for a nearest-class assignment to newly discovered content types and thus fit our requirements. However, for a large number of classes, call it $k$, a hierarchical partitioning also allows for better lookup times ($O(\log(k))$ vs $O(k)$) as demonstrated in Section~\ref{sec:purity}. To offset the compute intensive build times of hierarchical structures (compared to, say, k-means) we have implemented a solution in Spark that scales to millions of data points and thousands of classes in a production environment (see Section~\ref{sec:scale}). 

Our proposed solution has the following features. 

    (1) An eigenvector based algorithm that has guaranteed approximations to the hierarchy cost. Note that given the inherent nature of spectral approximations ($\sqrt{OPT}$ instead of a constant factor approximation to the sparsest cut problem in graphs where $OPT$ denotes the optimal) proving such a guarantee is non-trivial. On the plus side, using approximate eigenvectors enables a low latency, high throughput classification of new content into hundreds, possibly thousands of classes to incorporate a rich content taxonomy. 
   
    (2) Ability to detect new content types and human-friendly ways of identifying and labeling new content, to detect anomalous growth in any content category, which may indicate a trending event or a novel spam attack, as well as ability to fall back on surrogate categories for rapid response in such situations.
    
    (3) Robustness to overlapping categories and the ability for a human editor to explicitly incorporate domain knowledge to enforce separation of overlapping categories. For example, promotional content and scams often contain very similar text, but they require very different responses: scams can be harmful to the recipient and must be removed from site while promotional content may simply receive diminished distribution. The use of a hierarchy enables handling overlapping categories since such category pairs will likely be separated deeper in the hierarchy. 
    
    (4) Ability to work with both labeled and unlabeled data: given the high volume of user generated content, it is prohibitive to obtain an up-to-date and comprehensively labeled dataset. For identifying spam, positive examples are rare. Nevertheless, the site still needs to respond to newly arriving unlabeled content clusters. 

Note that once we build the hierarchy, the resulting multi-class classification could take advantage of 
labels from heterogeneous sources, thereby incorporating all the existing pipelines of label collection. 

\section{Related Work}
\label{sec:relatedwork}

While a number of clustering schemes have been presented in the literature, no hierarchy sensitive cost function has been available for assessing clustering quality until recently. Dasgupta~\cite{Dasgupta:2016} proposed an objective function that is sensitive to both, the pairwise similarity as well as the position of clusters within a hierarchy. We leverage this cost function to bound worst-case costs for our  scheme for three different splitting criteria. Dasgupta's work was followed by an axiomatic development of a family of objective functions~\cite{DBLP:journals/corr/Cohen-AddadKMM17} that have the desirable properties of a hierarchical cost function. However, to the best of our knowledge, ours is the first work that shows the usefulness of the proposed hierarchical objective function in a concrete real world task by relating these costs to a set of web-scale clustering methods. 

While there are other approximation algorithms for hierarchical clustering~\cite{Dasgupta:2016}, these algorithms are based on max-flow\cite{leighton1999multicommodity} and semi-definite programming~\cite{arora2009expander}. These methods are not appropriate in our setting because of two reasons: first, these algorithms do not scale to our needs. Second, it is difficult to use these algorithms incrementally, as required by our anomaly/trend detection applications. Traditional, bottom-up, hierarchical clustering methods such as linkages\cite{Day1984}, do not scale to our application as they are all memory and time intensive both, at build as well as at run time.

Nearest neighbor search algorithms, such as k-d trees\cite{Bentley:1975:MBS:361002.361007} and Locality Sensitive Hashing (LSH)\cite{Indyk:1998:ANN:276698.276876} are not suitable to our application either. The k-d tree algorithm is not suitable for high dimensional and sparse data such as ours. LSH does not have a hierarchical structure which is a core requirement for us, as described above. Algorithms such as LOM Trees \cite{choromanska2015logarithmic} and Recall Trees\cite{daume2017logarithmic} tackle the logarithmic time extreme classification problem. However, these are supervised classification approaches and are unsuitable for our application which requires both, unsupervised as well as supervised modes. 

Spatial data structures for high dimensional data have also been studied in the literature, but not in a classification setting. Our data structures are most similar to Random Projection Trees\cite{dasgupta2008random}, PCA Trees\cite{verma2009spatial}, from which we differ as we use the second eigenvector with the aim of getting close to the sparsest cut, and APD Trees \cite{DBLP:journals/corr/abs-1206-4668}. Again, while these works focus on showing the diameter reduction at each step, we focus on the hierarchical cost function~\cite{Dasgupta:2016} and performance on classification and anomaly detection tasks.

\section{Proposed Solution}
We propose a \textit{top-down} algorithm for creating a hierarchical partitioning as shown in Algorithm~\ref{alg:metahierarchy}. This algorithm relies on a function call \textsc{splittingrule} that essentially determines the strategy to create the hierarchy. We consider the whole dataset, find the best split based on the \textsc{splittingrule} and recursively partition the data till we have sets with cardinality one.

\begin{algorithm}[t]
\caption{Hierarchical Clustering}\label{alg:metahierarchy}
\begin{algorithmic}[1]
\Procedure{MakeTree}{$V$}
\If{$|V| = 1$}
    \State {\textbf{return} leaf containing $|V|$}
\EndIf

\State $h_r, s_r \gets $\textsc{splittingRule}$(V)$
\State $S \gets \{x \in V, h_r \cdot x \le s_r\}$
\State LeftTree $\gets$ MakeTree(S)
\State RightTree  $\gets$ MakeTree($V\setminus S$)
\State \textbf{return} [$h_r$, $s_r$, LeftTree, RightTree]

\EndProcedure

\Procedure{Query}{$T, x$}
\State $r=$root$(T)$.
\If {$T$ contains  $ < B$ points \textbf{or} $r$ has no children}         \State return all of the points.
\Else
    \State $h_r, s_r$ be splitting hyperplane, location for $r$.
    \If{ $x \cdot h_r \le s_r$}  
        \State{Query left child of $r$ with $x$}
    \Else 
        \State{Query right child of $r$ with $x$}
    \EndIf
\EndIf
\EndProcedure

\Procedure{Classify}{$T, x, k$}
\State $S \leftarrow$\textsc{Query}$(T,x)$
\State $C \gets$ $k$ nearest neighbors to $x$ from $S$.
\State $l_x \leftarrow$ majority vote of the labels of $C$.
\State \textbf{return} $l_x$
\EndProcedure
\end{algorithmic}
\end{algorithm}

\begin{algorithm}[t]
\caption{Eigenvector partitioning}\label{alg:eigen}
\begin{algorithmic}[1]
\Procedure{splittingrule}{$V$}
\If{$|V| = 1$}
    \State {\textbf{return} leaf containing $|V|$}
\EndIf

\State Let $A\in \Re^{\ell\times d}$ be the matrix formed by the points.
\State $d = A e$ where $e$ is vector of all ones, $D = \mathsf{diag}(d).$
\State $\tilde{A} = D^{-1/2}A$.
\State $h_t\in \Re^d$ be the second largest right singular vector of $\tilde{A}$.
\State Create $X = \{\tilde{A}_{i*} \cdot h_t,\ \forall i\in [\ell]  \}$.
\State Consider the sorted sequence $X = \langle X_1,\ldots,X_{\ell}\rangle$. \\
Find the 
set $S_j = \{k\in \ell, X_k \le X_j\}$ that has smallest $\gamma(S_j)$. 

\State \textbf{return} $(h_t, X_j)$

\EndProcedure
\end{algorithmic}
\end{algorithm}

\begin{algorithm}
\label{alg:approxevec}
\caption{Approximate Eigenvector}
\label{alg:APD}
\begin{algorithmic}[1]
\Procedure{splittingrule}{V}
\State Let $A\in \Re^{\ell\times d}$ be the matrix formed by $\{v_1, \ldots,v_{\ell}\}$.
\State $v\leftarrow $ random vector in $\Re^d$.
\State $d = A e$ where $e$ is vector of all ones, $D = \mathsf{diag}(d).$
\State $\tilde{A} = D^{-1/2}A$.
\State \% $M_1, M_2$ are $O(\log(n) / \epsilon)$.
\For {$i \in 1\ldots M_1$}
\State $v \leftarrow \tilde{A}^t(\tilde{A}v)$, $v\leftarrow v/\|v\|_2$
\EndFor
\State $u \leftarrow $ random vector in $\Re^d$.
\For{$i \in 1\ldots M_2$}
\State $x \leftarrow (I - vv^t)\tilde{A}^t\tilde{A}x$.
\State $x \leftarrow x/\|x\|_2$
\EndFor
\State Create $X = \{\tilde{A}_{i*} \cdot x,\ \forall i\in [\ell]  \}$.
\State Consider the sorted sequence $X = \langle X_1,\ldots,X_{\ell}\rangle$. \\
Find the 
set $S_j = \{k\in \ell, X_k \le X_j\}$ that has smallest $\gamma(S_j)$. 

\State \textbf{return} $(x, X_j)$.
\EndProcedure
\end{algorithmic}
\end{algorithm}

\begin{algorithm}[htb]
\caption{$2$-means}
\label{alg:2means}
\begin{algorithmic}[1]
\Procedure{SplittingRule}{V}
\State Let $A\in \Re^{\ell\times d}$ be the matrix formed by $\{v_1, \ldots,v_{\ell}\}$.
\State Use k-means++~\cite{arthur2007k} to get a bi-partitioning of $V$.
\State Let $c_1$ and $c_2$ be the two cluster centers.
\State \textbf{return} $h_t = 2(c_1 -c_2)$, $s_t \leftarrow \|c_1^2\| - \|c_2\|^2$.
\EndProcedure
\end{algorithmic}
\end{algorithm}
\begin{algorithm}[htbp]
\caption{Random Partitioning}
\label{alg:random}
\begin{algorithmic}[1]
\Procedure{SplittingRule}{V}
\State $h_t \sim N(0,I_{dxd})$, $s_t \leftarrow 0.$
\State \textbf{return} $h_t/\|h_t\|_2, s_t$.
\EndProcedure
\end{algorithmic}
\end{algorithm}

We now describe four splitting strategies for \textsc{splittingrule} used. 
Let $t$ denote a generic tree-node in the hierarchy, in which we have points $V_t = \{v_1, \ldots, v_{\ell}\}$.
Let $h_t$ and $s_t$ denote the splitting hyperplane and split point respectively, which are computed by the corresponding splitting strategy.  

\begin{enumerate}
    \item \texttt{Random Hyperplane (RP)}: We sample $h_t$ uniformly at random from $S^{d-1}= \{x \in \Re^d, \|x\|_2 = 1\}$. This can be done by sampling $y\sim \mathcal{N}(0, I_d)$ and returning $x = \frac{y}{\|y\|}$. 
    Define $X = \{h_t \cdot v_i, i = 1,\ldots, \ell\}$. The splitting point $s_t$ is chosen as the median of $X$.

    \item \texttt{Eigenvector (EV)}: Let $D\in \Re^{\ell\times \ell}$ be the diagonal matrix with $D_{ii} = \sum_j A_{ij}$. 
    We calculate the normalized feature-feature matrix $F = A^t D^{-1} A \in \Re^{d\times d}$. Let $x_{2}$ be the eigenvector corresponding to the second largest eigenvalue of $F$. 
    We define  $u_2 = D^{-1/2} A x_2/\| D^{-1/2} A x_2\|_2$. 
    Note that $u_{2}$ is the second largest singular vector of $D^{-1/2}A A^t D^{-1/2}$. The coordinates of $u_2$ are sorted. 
    We find the median, say $X_j$ and return the set $S_j = \{k \in \ell, X_k \leq X_j\}$ that has the smallest $\gamma(S_j)$. 
    The pair $(x_2, X_j)$ is returned as the splitting hyperplane and threshold.
    
    \item \texttt{Approximate Eigenvector (AEV)}: Inspired by McGregor et al.~\cite{DBLP:journals/corr/abs-1206-4668}, we use an  approximate eigenvector, rather than the exact one. 
    In Algorithm~\ref{alg:APD}, we present the pseudocode of how we use power iterations to estimate the eigenvector corresponding to the second largest eigenvalue of the normalized covariance matrix.
    
    \item \texttt{$2$-means}: We find $2$-means clustering of the data and define $h_t$ to be the vector that identifies the hyperplane equally distant from the two centers. 
    The threshold $s_t$ is set to be $0$, and the two clusters are returned as partitions. \footnote{Our conjecture is that $2$-means does not provide a worst-case multiplicative guarantee. \cite{moseley2017approximation} demonstrate this claim for a related hierarchical objective.}
\end{enumerate}

For each node in the hierarchy we maintain the splitting hyperplane and the split point, which enables us to query the hierarchy at the test time with individual query points.
We would also like to make a point that hierarchy is created in batch mode.

\section{Theoretical Guarantees}
\label{sec:theory} 

\subsection{Preliminaries}
Our input is a set of points (images or documents), represented as unit norm vectors in $d$-dimensional space. We represent the collection by a matrix $X \in \mbb{R}^{n\times d}$.

Let $ G= (V, E)$ be a symmetric weighted graph and $C$ be the adjacency matrix of $G$, i.e. $C_{ij}$
captures the weight of the edge $\{i, j\}$, where the weights $C_{ij}$ are meant to denote \emph{similarity}, not distance, and $0\le C_{ij}\le 1$. In our case, given a normalized representation of $n$ objects as the rows of the matrix $X$, we will implicitly calculate $C$ as $C = XX^t$ (\textbf{Note:} Our algorithms never actually calculate the matrix $C$, as that is computationally expensive). The (weighted) degree of each node $i$ is defined as $d_i = \sum_{j} C_{ij}$. Given a subset $S$ of vertices, we will refer to the degree or volume of $S$ as $d(S) = \sum_{i\in S} d_i$. The size of the cut $(S, V\setminus S)$ is defined as $C(S, V\setminus S) = \sum_{i\in S, j\in V\setminus S} C_{ij}$.  
We refer to the following standard definitions and results\cite{trevisan2013lecture}:

Let $(S,V\setminus S)$ be a partition of the vertices (i.e. a cut).  The expansion of the cut is
defined as $\phi(S) := \frac{C(S, V\setminus S)}{\min(|S|,|V\setminus S|)}.$ The expansion of the entire graph $G$ is defined as 
$\phi(G) = \min_S \phi(S)$. Similarly, the conductance of the cut $(S,V-S)$ is
$ \gamma(S) = \frac{C(S, V\setminus S)}{\min(d(S), d(V\setminus S))},$ and $\gamma(G) = \min_S \gamma(S)$.
    
When the degrees are very irregular, it is natural to work with the normalized adjacency matrix which is defined as follows: let $D$ be a diagonal matrix such that $D_{ii}= d_i$, and $\tilde{C} = D^{-1/2}CD^{-1/2}$. It is clear that the first
eigenvector of $\tilde{C}$ is proportional to$(\sqrt{d_1}, \ldots, \sqrt{d_n})$ and the first eigenvalue is $1$~\cite{trevisan2013lecture}. 
Let $\lambda_2$ be the second largest eigenvalue of $\tilde{C}$ and $v_2$ the corresponding eigenvector. We quote the following celebrated result about the gap $1 - \lambda_2$. Note that while it is traditional to work with the Laplacian matrix and the smallest eigenvalues, we will state and apply Cheeger's inequality for the second largest eigenvalue of the normalized covariance matrix. This is because we will approximate the eigenvector using power iterations and, it is easier to approximate the second largest eigenvector than the second smallest. 

\begin{theorem}\cite{MOHAR1989274} (\textbf{Cheeger's Inequalities}) $ \frac{1-\lambda_2}{2} \leq \gamma(G) \leq \sqrt{2 (1-\lambda_2)}$.
%
Also, denote the sorted coordinates of the second smallest eigenvector of the Laplacian  $x_1\le x_2\ldots \le x_n$, and define the $n$ cuts $S_j$ as $S_j = \{ i, x_i \le x_j\}$. Then one of the cuts $S_j$ has a conductance at most $O(\sqrt{\gamma(G)})$.
\label{thm:cheeger}
\end{theorem}

The following proposition is easy to prove using $\Delta_{min} = \min_i\: d_i$, $\Delta_{max} = \max_i\: d_i$ and replacing $\min(d(S),d(V\setminus S))$ by $\Delta_{min} \min(|S|,$ $|V \setminus S|)$ and $\Delta_{max} \min(|S|, |V \setminus S|)$.

\begin{lemma}
If $C\in \Re^{n\times n}$ is such that $\Delta_{min} \le d_i\le \Delta_{max}$, then $\forall$ set $S$, $\frac{\phi(S)}{\Delta_{max}}\le \gamma(S) \le \frac{\phi(S)}{\Delta_{min}}$.   
\end{lemma}

\subsubsection{Hierarchical Clustering Cost}

Consider a node-set $V$ and $C = \{C_{ij}, i,j \in V\}$ be defined as before. 
Let $T$ be a rooted hierarchical clustering tree (not necessarily binary) for graph $G$, where the leaves of the tree $T$ are the nodes in the graph $G$. 
For any two nodes $i$ and $j$ suppose $i \lor j$ denotes the lowest common ancestor in $T$ of leaves $i$ and $j$, $T[i]$ denote the subtree rooted at (tree) node $i$ and $leaves(T[i])$ denote the set of leaves of the subtree $T[i]$. Then, the cost of tree $T$ is defined as follows \cite{Dasgupta:2016}:
\begin{align}
\label{eq:hiercost}
    cost_{V}(T) = \sum_{\{i,j\} \in V\times V} C_{ij} \times |leaves(T[i \lor j])|  
\end{align}

Given a set of points $V$ and similarities $C_{ij}$ among them, the aim is to return a tree that achieves minimum value of the objective in equation~\ref{eq:hiercost}. The above cost function has been justified by various axiomatic approaches~\cite{Dasgupta:2016,DBLP:journals/corr/Cohen-AddadKMM17}, and approximation algorithms~\cite{Charikar:2017:AHC:3039686.3039739}, \cite{monathgradient}, \cite{DBLP:journals/corr/RoyP16a} have been developed for it. We next show that \texttt{EV}, \texttt{AEV}, \texttt{RP} splitting strategies give approximations to the above cost function. 

Our algorithm for the hierarchy cost utilizes a number of instances of
sparsest cut solutions. While Arora et. al (ARV) give a $O(\sqrt{\log n})$ approximation to the sparsest cut, spectral methods (which we employ because of ``queryability'' of the structure) only give guarantees of the form $\sqrt{OPT}$ \cite{trevisan2013lecture}. This implies that the guarantee is better when $OPT$ is higher, e.g. constant for expander graphs. We make use of exact and approximate eigenvectors to perform recursive spectral partitioning in \texttt{EV} and \texttt{AEV} methods, and show that a similar guarantee also holds for the hierarchy cost. The $O(\sqrt{OPT})$ acts as a boon as cost values are typically high due to the $|leaves(i \lor j)|$ factor in the cost function. For our setting $OPT=\Omega(n^3 \min_{u,v} dist(u,v))$, e.g. for unit weighted cliques, it is $\Omega(n^3)$. We note that i) our analysis is worst case and possibly not tight, ii) our algorithms are practical to implement unlike SDP and flow and, iii) our experiments show that the hierarchy output by our algorithms have good performance for the tasks that we posed.

\subsection{Using Eigenvectors}
We first show that our partitioning strategy of using second left singular vector of the matrix $V\in \Re^{\ell \times d}$ gives an approximation to the  objective~\eqref{eq:hiercost}. In order to do this, we quote a particular lemma from~\cite{Dasgupta:2016} that characterizes the structure of any tree, in particular the optimal tree. 

\begin{lemma}
\label{lemma:dasgupta}
\label{L:1} (Lemma 11 of \cite{Dasgupta:2016}):
Pick any binary tree, $T$ on $V$. There exists a partition $A,B$ of $V$, where $\frac{|V|}{3} \leq |A|, |B| \leq \frac{2|V|}{3}$ such that $\frac{C(A,B)}{|A||B|} < \frac{27}{4|V|^3}cost_V(T)$ where $C(A,B) = \sum_{i\in A, j \in B} C_{ij}$, $C_{ij}$ is the weight of the edge between $i$ and $j$, and $cost_V(T)$ refers to the cost of the binary tree on $V$.
\end{lemma}

We show a bound on the cost of the tree formed by using exact eigenvector as the splitting subroutine in Algorithm~\ref{alg:metahierarchy}.

\begin{theorem}
\label{thm:evec}
Let $V$ contain $n$ points. The point-point similarities are encoded as $C : V\times V \rightarrow [0,1]$. 
Let tree $T^*$ be a minimizer of $cost_V(\cdot)$ and let $T$ be the tree returned by using the second largest eigenvector of $D^{-1/2}VV^tD^{-1/2}$ as a splitting rule. Then 
\begin{math} cost_V(T) \leq c n\;\log(n)\; \sqrt{cost_V(T^*)}\end{math} 
for some $c \le {\Delta_{max} \sqrt{\frac{27}{\Delta_{min}}}}$, where $\Delta_{min} = \min_{i \in V} deg(i)$ and $\Delta_{max} = \max_{i \in V} deg(i)$.
\end{theorem}

\begin{proof}
We prove this by induction. The base case, when $n=1$, the $cost(T)=0$.
Let us assume that the claim holds for graphs with at most $n-1$ nodes.

Suppose the optimal tree is $T^*$ and consider the root of $T^*$, with the entire set of points $V$. From Lemma~\ref{lemma:dasgupta}, we can say that there exists a partition $(A,B)$ of $V$ such that $\frac{C(A,B)}{|A||B|}$ will be bounded by $\frac{27 cost(T^*)}{4n^3}.$
 When an exact eigenvector is used as the splitting strategy, for a node $t$, if $\mathcal{V}$ is the data matrix which contains all the points as the rows, and $h_t$ is the second right singular vector of $D^{-1/2}\mathcal{V}$, then $D^{-1/2}\mathcal{V}\cdot h_t$ is the (scaled) second left singular vector of $D^{-1/2}\mathcal{V}$, as well as the second largest eigenvector of $D^{-1/2}\mathcal{V}\mathcal{V}^tD^{-1/2} = D^{-1/2}CD^{-1/2}$. Let $(A,B)$ (say $|A|\le |A\cup B|/2$) denote the spectral cut found. Using Lemma~\ref{thm:cheeger}, we know that $\gamma(A) \le O(\sqrt{\gamma(C)})$, where $\gamma(C)$ denotes the conductance of $C$. Given that the degrees of $C$ all lie in the range $[\Delta_{min},\Delta_{max}]$, we know that $\phi(C)\ge \Delta_{min}\cdot\gamma(C)$, and $\phi(A) \le \Delta_{max}\cdot \gamma(A)$. Hence,
 \begin{align*}
 \phi(A) \le  {\Delta_{max}}\cdot \gamma(A) 
 \\
 \phi(A) \le {\Delta_{max}}\sqrt{\gamma(C)}  \le  \frac{\Delta_{max}}{\sqrt{\Delta_{min}}} \sqrt{\phi(C)} \end{align*}
  
Note that for any partitioning $(S, S^c)$, 
$\frac{1}{n}\phi(S)  \le \frac{C(S,S^c)}{|S||S^c|}\le \frac{2}{n}\phi(S).$
Suppose $(A', B')$ is the split guaranteed by Lemma~\ref{L:1}. Recall that $(A, B)$ is the split returned by the eigenvector. Let $\theta = \frac{\Delta_{max}}{\sqrt{\Delta_{min}}}$. 
Then, using Lemma \ref{L:1}

\begin{align*}
    \frac{C(A,B)}{|A||B|}& \le \frac{2\phi(A,B)}{n} \le  \frac{\theta\sqrt{\phi(A',B')}}{n} \le {\theta}\sqrt{\frac{C(A',B')}{n|A'||B'|}} \\ 
    & \le \frac{\theta}{\sqrt{n}}\sqrt{\frac{27}{4|V|^3} cost_V(T^*)}=\theta \sqrt{\frac{27 cost_V(T^*)}{4|V|^4}}
\end{align*}
$$\frac{n C(A,B)}{|A||B|}\le \theta \sqrt{\frac{27 cost_V(T^*)} {4|V|^2}}$$

Let us consider trees $T^*_A, T_A$ and $T^*_B, T_B$ which are the trees $T^*$ and $T$ restricted to the nodes in $A$ and $B$ respectively. Recall that $T^*$ is the optimal tree.

So,
Let $0 < |A| = pn \leq \frac{n}{2}$. Let $|A| \leq |B| = (1-p)n$. Therefore,
$ nC(A,B) \le p(1-p)n \theta \sqrt{\frac{27 cost_V(T^*)}{4}}$.
Since $|A|,|B| < n$, we apply the induction hypothesis.

\begin{align*}
    cost_A(T_A) + cost_B(T_B)  \leq & 
    n\theta \left( p\sqrt{ \frac{27}{4} cost_A(T^*_A)}\log(pn) 
    +(1-p) \sqrt{\frac{27}{4} cost_B(T^*_B)}\log((1-p)n)\right)\\
    \leq &n\theta\sqrt{\frac{27 cost_V(T^*)}{4}}(p\log p+(1-p)\log(1-p) + \log n)
\end{align*}

Now,
    $cost_V(T) \leq n C(A,B) + cost_A(T_A) + cost_B(T_B) 
            \leq n \theta \frac{\sqrt{27cost_V(T^*)}}{2} (p(1-p) + 
             p\log p + (1-p)\log(1-p) + \log n) $
            $\leq n \log n \theta \frac{\sqrt{27 cost(T^*)}}{2}.$
\end{proof}

\paragraph{Intuition behind the bound.} The above bound is not in the form of an approximation bound expressible as $cost_V(T) \le f(n)  cost_V(T^*)$
for some function of $n$. Recall that the core of the eigenvector algorithm is to use eigenvector based partitioning to create low conductance cuts. 
Spectral methods, when used to find low-conductance cuts, only give guarantees of the form $\sqrt{OPT}$, not $f(n) \times OPT$, which stems from 
Cheeger's inequality (Theorem~\ref{thm:cheeger}). It is well known that this guarantee for spectral cuts is tight, attained for instance by cycle graphs. Furthermore, for a weighted clique with all weights to be $1$, we can show that  for any hierarchy $T$, $cost_V(T) = n^3$. What we get using the bound, by plugging in $\Delta_{\max} = \Delta_{\min} = n$ is $cost_V(T) \leq O(n^3 \log n$). This is the same as the approximation obtained by using the LR\cite{leighton1999multicommodity} algorithm to build the hierarchy. 
\subsubsection{Approximation to the Planted Partition}
In this section, we show that the bound given by spectral partitioning is significantly better if the data does have a nice structure. Traditionally, merits of spectral partitioning have been theoretically demonstrated using stochastic block models or mixture models. We consider a similar setting. Note that this is an idealized setting meant to demonstrate that spectral partitioning based splitting can give good approximation to the hierarchical cost under simple generative models.  

Assume that the similarity matrix $C$ is generated using a $(n,p,q)$-planted partition model with two equisized partitions. $C$ is constructed through the following generative process: for every pair of distinct nodes, $i,j$, 
an edge is added with probability $p$, 
if $i$ and $j$ belong to the same partition and $q$ otherwise.
All $\binom{n}{2}$ decisions are independent.

We use the following result that shows that given $C$, the eigenvector based sparsest cut algorithm would find a partition close to the ground truth. 
\begin{lemma}
\label{lem:mcsherry}
From \cite{mcsherry2001spectral}: In the planted partition model, the spectral sparsest cut algorithm mis-clusters at most a constant, $k$ number of points where $ k = \frac{36p}{(p-q)^2}$.
\end{lemma}

\begin{lemma}
\label{lemma:clique}
Let $\mathcal{G}(V,E)$ be an unweighted undirected clique on $n$ nodes, and let $\mathcal{G^{'}}(V^{'},E^{'})$ be another unweighted undirected clique on $n+k$ nodes, then 
$cost_{\mathcal{G}^{'}}(\cdot) = c\times cost_{\mathcal{G}}(\cdot)$ for $n >> 1$ and $k \leq n$ and $c$ is some constant.
\end{lemma}

\begin{proof}
From \cite{Dasgupta:2016}, we know that 
$cost_\mc{G}(T) = \frac{1}{3}(n^3-n)$ and $cost_{\mathcal{G}^{'}}(T) = \frac{1}{3}((n+k)^3 - n - k)$, and every tree has the same cost.

Now,

\begin{align*}
    \frac{cost_{\mathcal{G}^{'}}(T)}{cost_\mc{G}(T)} &= \frac{\frac{1}{3}((n+k)^3 - n - k)}{\frac{1}{3}(n^3-n)} 
    = \frac{(n^3 + k^3 + 3nk(n+k)) - n - k}{n^3 - n}\\
    &= \frac{n^3 (1 + \frac{k^3}{n^3} + 3\frac{k(n+k)}{n^2} - \frac{1}{n^2} - \frac{k}{n^3})}{n^3(1 - \frac{1}{n^2})}
    = \frac{(1+1+3)}{1} \quad \quad (Since \: k \leq n)\\
    &= (1 + \epsilon) cost_{\mathcal{G}^{'}}(T) = (c)cost_\mc{G}(T)
\end{align*}

\end{proof}

\begin{theorem}
    Let the matrix $C$ be generated from the above planted partition model. Let $V$ contain $n$ points and let $T^*$ be a minimizer of $cost_V(\cdot)$ and if $T$ is the tree returned by using the splitting rule of \texttt{EV}~ inside Algorithm~\ref{alg:metahierarchy}, then $cost_V(T) \leq c_1 cost_V(T^*)$ where $c_1 = max\left( 1 + \frac{72p}{qn(p-q)}, c\right)$ and $c$ is a constant.
\end{theorem}

\begin{proof}

Let $d$ be the expected degree of $\mc{G}$. Then, $d = (p+q)\frac{n}{2} - p \approx (p+q)\frac{n}{2}$. Since $\mc{G}$ is an almost regular graph, we do not perform the degree normalization for this theorem.

We quote a small result on the first eigenvector of $C$ and its corresponding $\mc{L} = I - D^{-1/2}CD^{-1/2}$ where $D$ is the diagonal matrix with $d_{ii} = d$. The first eigenvectors of $C, \mc{L}$ are given as

\[C\mbf{1} = \frac{n}{2}(p+q)\mbf{1} \quad \text{and}
 \quad \mc{L}\mbf{v} = 1 \mbf{v} \quad \quad (\mbf{v} = d^{-1/2}\mbf{1})\]

Let $T$ be the tree that we get using Algorithm~\ref{alg:eigen} inside Algorithm~\ref{alg:metahierarchy}, and $T^*$ be the optimal tree. For the planted partitions, the expected optimal cost is given by
\begin{align*}
    \mbb{E}cost_V(T^*) &=  \sum_{i,j}\mbb{E}w_{ij} |leaves(T^*[i \lor j])| \\
    & = \sum_{i,j} Prob[edge \: (i,j)]\: |leaves(T^*[i \lor j])| 
    \\ & = \frac{qn^3}{4} + 2\frac{p}{3}\left(\frac{n^3}{8} - \frac{n}{2}\right)
\end{align*}
Using Lemma~\ref{lem:mcsherry}, we misclassify at most $k = \frac{36p}{(p-q)^2}$ vertices, the cost of the top split in $T$ is

\begin{align*}
    \mbb{E}cost_V(T) &\leq \frac{qn^3}{4} + \frac{36p}{2(p-q)^2}pn^2 - \frac{36p}{2(p-q)^2}qn^2 \\
   & = \frac{qn^3}{4} + \frac{36pn^2}{2(p-q)} 
    = \frac{qn^3}{4}\left(1 + \frac{72p}{qn(p-q)}\right)
\end{align*}

For the future splits, we make two assumptions
\begin{enumerate}
    \item Some edges will have weight $C_{ij} < p$, and
    \item The size of both the partitions will be in the range $[\frac{n}{2} - k, \frac{n}{2}+k]$, and $k = \frac{36p}{(p-q)^2}$.
\end{enumerate}

Also, since we have two equi-sized clusters, $k \leq \frac{n}{2}$. Now, let the two partitions identified by $A,B$. Using Lemma~\ref{lemma:clique},

\begin{align*}
    \mathbb{E} cost_A(T_A) + \mbb{E}cost_B(T_B) & \leq \frac{2p}{3}c \left(\frac{n^3}{8} - \frac{n}{2}\right)
\end{align*}

Thus, the total cost can be bounded by

\begin{align*}
    \mathbb{E}cost_V(T) & \leq \frac{qn^3}{4}\left(1 + \frac{72p}{qn(p-q)}\right) + c\frac{2p}{3}\left(\frac{n^3}{8} - \frac{n}{2}\right)\\
    &\leq c_1 \left(\frac{qn^3}{4} + \frac{2p}{3}\left(\frac{n^3}{8} - \frac{n}{2}\right)\right)
    \quad \quad \left( where\: c_1 = max\left(1+ \frac{72p}{qn(p-q)}, c\right) \right) \\
    &= c_1 \mbb{E}cost_V(T^*)
\end{align*}
\end{proof}

While we state the above result for simple planted partition model, it is possible to come up with an analogous statement for other variants, e.g. when the points come from a Gaussian mixture model. 

\subsection{Using Approximate Eigenvectors}

Algorithm~\ref{alg:APD} starts with a random vector and applies power iterations to the vector. We thus have a close approximation to the first eigenvector, using which we get a close approximation to the second, again by applying power iteration. We first quote a result that shows that assuming we have approximated the first singular vector $v_1$, the second singular vector of $\tilde{C}$ is also well approximated. The proof of the main theorem follows that of Theorem~\ref{thm:evec}.

\begin{theorem}(Theorem 4.6 in \cite{trevisan2013lecture})
\label{th:approx_lambda_2_luca_lb}
Let $M\in \Re^{d\times d}$ be a PSD matrix, $k$ a positive integer and $\epsilon > 0$. If $\mathbf{v}_d$ is the eigenvector corresponding largest eigenvalue of M, then with probability $> 3/16$ over the choice of $x_0$, the power iteration algorithm outputs a vector $x \perp \mathbf{v}_1$, such that
$\frac{{x}^tM{x}}{{x}^t{x}}  \geq  \frac{\lambda_2 (1-\epsilon)}{1 + 4n(1-\epsilon)^{2k}}.$
\end{theorem}

\begin{lemma}Consider a matrix $M = I - L = D^{-1/2}XX^tD^{-1/2}$ and $0= \lambda_1 \leq ... \leq \lambda_n \leq 2$ are the eigenvalues of $L$, the laplacian matrix. By running power iteration on $M' = X^tD^{-1}X$, we get
\begin{math}
\label{eq:approx_lambda}
    \frac{{x}^tL{x}}{{x}^t{x}} \leq \lambda_2 + \epsilon
\end{math}
\end{lemma}

\begin{proof}

If $0= \lambda_1 \leq \ldots \leq \lambda_n \leq 2$ are the eigenvalues of $L$, $1 = 1- \lambda_1 \geq \ldots \geq 1 - \lambda_n \geq -1$ are the eigenvalues of $M$.

Now, $M = D^{-1/2}XX^tD^{-1/2} = {(X^tD^{-1/2})}^t{(X^tD^{-1/2})}$. Note that $M$ is a PSD matrix, and let $M = Y^tY$, where $Y = (A^tD^{-1/2})$. Also, eigenvalues for $M$ will be bounded in the range $[0,1]$.
Also, eigenvalues for $Y^tY$ are the same as the eigenvalues for $YY^t$. And, $YY^t = X^tD^{-1}X = M'$. This also implies that, since we use $C = XX^t$, the eigenvalues of $L$ are also in the range $[0,1]$.

Thus, we prove the inequality~\ref{eq:approx_lambda} using $M$ instead of $M'$. Using Theorem~\ref{th:approx_lambda_2_luca_lb}, using power method and setting $k = O(\frac{\log n}{\epsilon})$, we have

\begin{align*}
    \frac{x^tMx}{x^tx} \geq (1 - \lambda_2)(1 - \epsilon), \qquad \quad \quad 
    \frac{x^tLx}{x^tx} \leq \lambda_2 + \epsilon
\end{align*}
\end{proof}

\begin{theorem}
\label{thm:approxev}
Let $V$ contain $n$ points. The point-point similarities are encoded as $C : V\times V \rightarrow [0,1]$. 
Let tree $T^*$ be a minimizer of $cost_V(\cdot)$ and let $T$ be the tree returned by using the Algorithm~\ref{alg:metahierarchy} with the splitting rule defined in Algorithm~\ref{alg:APD}. Then 
\begin{math} cost_V(T) \leq c n\;\log(n)\; \sqrt{cost_V(T^*) + \epsilon}\end{math} 
for some $c \le {\Delta_{max} \sqrt{\frac{27}{\Delta_{min}}}}$, where $\Delta_{min} = \min_{i \in V} deg(i)$ and $\Delta_{max} = \max_{i \in V} deg(i)$ and $\epsilon > 0$ is a small quantity.
\end{theorem}

\subsection{Using Random Hyperplanes}

We show that the expected cost of using random hyperplanes to do a clustering has a provable guarantee. 

\begin{theorem}
\label{thm:rp}
For any arbitrary dot product similarity matrix $C$, with all $0\le C_{ij}\le 1$, the hierarchy produced by random hyperplane based partitioning has expected approximation of $O(n)$.
\end{theorem}

\begin{proof}
Let the angle between points $i$ and $j$ be denoted by $ \theta_{ij} = \cos^{-1}(C_{ij})$. Given that the partitioning vector or equivalently, the hyperplane, has a direction which is uniformly chosen, and the points $i$ and $j$ are split iff they lie on different sides of the resulting hyperplane, at any given tree node that has both $i$ and $j$, $P(i,j \: split) = \frac{\theta_{ij}}{\pi}$.
and $P(i,j \: not \: split) = 1 - \frac{\theta_{ij}}{\pi}.$
Hence,
\begin{align*}
    P(i,j  \text{ split  at  level } l) 
    &= P(i,j \text{ not split till level }l-1) \cdot P(i,j \text{ split at level } l)
   \\ &= \left(1 - \frac{\theta_{ij}}{\pi}\right)^{(l-1)}\frac{\theta_{ij}}{\pi}
\end{align*}
If we use a median based splitting strategy, then each node at level $l$ has $n/2^l$ leaves. 
Hence if the pair $(i,j)$ are split at level $l$, then the contribution of the edge is $nC_{ij}/{2^l}$. 
Thus the expected cost using a median splitting strategy is:
\begin{align*}
    \mathbb{E}~cost_V(T) & \leq \sum_{i,j} nC_{ij} \sum_{l=1}^{\log(n)}\left(1 - \frac{\theta_{ij}}{\pi}\right)^{(l-1)}\frac{\theta_{ij}}{\pi}{\left(\frac{1}{2}\right)}^{l}
\end{align*}
The term inside the second summation is a GP.
\begin{align*}
    \mathbb{E}~cost_V(T) 
    \leq \sum_{i,j} 
    \frac{n}{2}C_{ij}
    \frac{\theta_{ij}}{\pi} 
    \left(\frac{1 - \left(\frac{1}{2}\left(1 - \frac{\theta_{ij}}{\pi}\right)\right)^{(\log(n)-1)}}{1 - \frac{1}{2}\left(1 - \frac{\theta_{ij}}{\pi}\right)}\right) \\
    \leq \sum_{i,j} \frac{n}{2}C_{ij} \frac{\theta_{ij}}{\pi}2 \leq \sum_{i,j} nC_{ij} 
    \quad (As \: 0 \leq \theta_{ij} \leq \pi)
\end{align*}

Now, we lower bound the optimal cost, $cost_V(T^*)$. There can be two cases, (1) $C_{ij} = 0$, and (2) $C_{ij} > 0$. For the second case $|leaves(T[i \lor j])| >= 2$. So, $cost_V(T^*) \geq 2\sum_{ij} C_{ij}$, where $T^*$ is the optimal tree. Thus,

$$\frac{\mathbb{E}cost_{V}(T)}{cost_{V}(T^*)} \leq \frac{\sum_{i,j} nC_{ij}}{ \sum_{i,j} 2C_{ij}} = \frac{n\sum_{i,j}C_{i,j}}{2 \sum_{i,j}C_{i,j}} = \frac{n}{2}$$

\end{proof}

\subsection{Time Complexities v/s Cost}
\label{sec:hccost}
In this section, we compare the time complexities v/s cost for various hierarchical clustering algorithms. We propose a hierarchical clustering algorithm with 4 splitting rules, namely \EV, \ASPD, \texttt{RP} and \texttt{2-means}. For \EV and \ASPD, the time to compute the eigenvector (or approximate eigenvector) is $O(nd^2)$ when $d \leq n$. The time to perform \texttt{RP} and \texttt{2-means} is $O(nd)$. Since we build a binary tree, the total building time is $O(nd^2\log n)$ for \EV and \ASPD. Also, the time for \texttt{2-means} and \texttt{RP} based hierarchical clustering is $O(nd\log n)$. The cost for \EV is given by Theorem~\ref{thm:evec}, for \ASPD is given by Theorem~\ref{thm:approxev} and for \texttt{RP} by Theorem~\ref{thm:rp}. Note that we do not provide any cost bounds for \texttt{2-means} based partitioning.
Though \texttt{ARV} and \texttt{LR} provide exciting bounds~\cite{Dasgupta:2016}, their build time is quite high due to the use of semi-definite programming and flows. This makes them impractical for our use case. Similarly, linkages do not scale well to large data as shown in Section~\ref{sec:evalNN}. The build times, query times and cost are summarized in Table~\ref{tab:quality}. \texttt{EV}, \texttt{AEV}, \texttt{RP}, and \texttt{2-means} have query time of $O(d\log n)$, and \texttt{Linkages}, \texttt{ARV} and \texttt{LR} based algorithms are not queryable.

\begin{table}[htb]
    \centering
    \caption{Comparison of Hierarchical Clustering Times and Quality. $\left(c \le {\Delta_{max} \sqrt{{27}/{\Delta_{min}}}}\right)$.}
    \label{tab:quality}
   
    \begin{tabular}{|c|c|c|c@{}|}
        \hline
        & \textbf{Build Time} & \textbf{Query Time} & \textbf{Cost} \\
        \hline
         \textbf{EV} & $O(nd^2 \log n)$ & $O(d\log n)$& $O(cn \log n \sqrt{OPT})$ \\
         \textbf{AEV} & $O(nd^2 \log n)$ & $O(d\log n)$& $O(cn \log n \sqrt{OPT + \epsilon})$ \\
         \textbf{RP} & $O(nd \log n)$ & $O(d\log n)$& $O(n OPT)$ \\
         \textbf{2-means} & $O(nd \log n)$ & $O(d\log n)$& - \\
         \hline
         \textbf{ARV} & $O(n^2\log^2n)$ & - & $O(\sqrt{\log n} OPT)$ \\
         \textbf{LR} & $O(n^3\log n)$ & - & $O({\log n} OPT)$ \\
         \textbf{Linkages} & $O(n^2\log n)$ & - & -\\
         \hline
    \end{tabular}
\end{table}

\section{Experiments}

\begin{figure*}
    \centering
    \includegraphics[width=\linewidth]{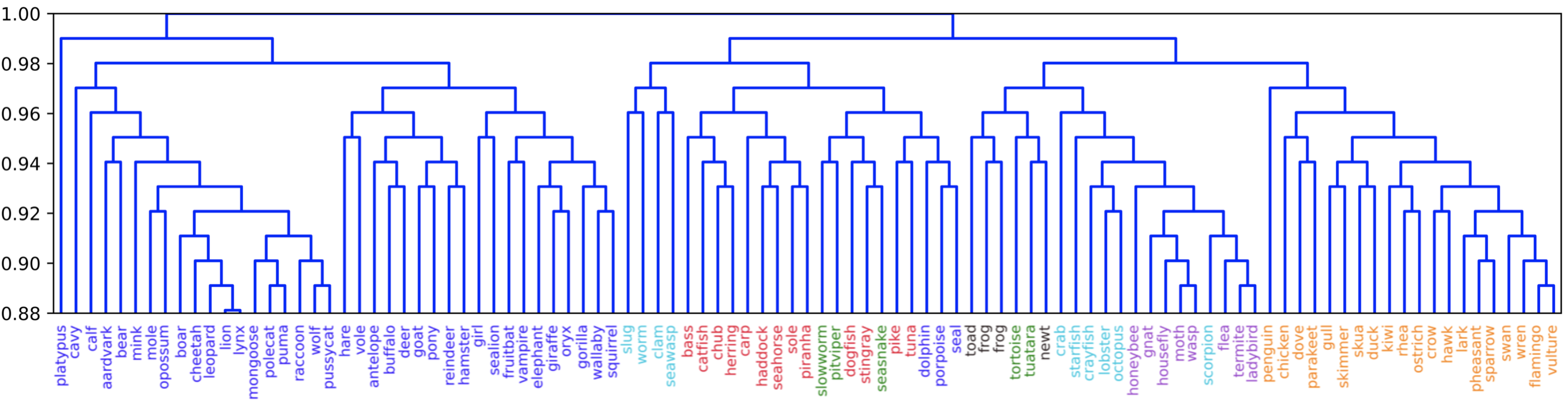}
    \caption{Hierarchy for Zoo dataset using Eigenvector based Hierarchical Clustering. The color of the label is an indicator of the ground truth class.}
    \label{fig:zoo}
\end{figure*}
In this section, we evaluate the performance of the our proposed solution (henceforth, Hierarchical Tree Based Nearest Neighbor,  or \HTNN) with four splitting rules described above. To benchmark our performance, we pick state of the art algorithms applied to three open source datasets and a spam dataset from one of the biggest social networking sites. The open source datasets are \ALOI \cite{rocha2014multiclass} (Dimensions-128, Classes=1k, Train size=88k, Test size=20k), \News\cite{Dua:2017} (Dimensions=100, Classes=20, Train size=11.3k, Test size = 7.5k), \CT\cite{Dua:2017} (GloVe embedding of dimensionality 100, Classes = 7, Train size = 480k, Test size = 100k). 

The spam dataset used is a \LComments dataset (GloVe embedding of dimensionality 50, Train Size = 100k, Test Size = 20k) which consists of text comments on the social networking website's home feed. This dataset has 3 labels, namely, \Spam (Content which is offensive, harmful, abusive or disruptive. These comments violate the Terms and Conditions of the networking site and are to be completely removed from the site), \texttt{Low-Quality} (\LQ)~ (Content which is irrelevant to the discussion or unappealing, and \Clear~ (Content which receives unrestricted distribution on the site). In our observations, \Spam~ and \LQ~ categories contain very similar content with fairly minor differences which leads to a difficult classification problem.

Since \LComments and \News\footnote{SVM on 20NG with ~20k dimensional tf-idf vectors gives f1-score as 92\%. However, it is infeasible to use such high dimensions for vector based strategies such as ours.} are textual data, each document is converted to a vector. Each word is converted using pretrained word GloVe embeddings\cite{pennington2014glove}. The average of these embeddings are taken to generate the document vector. 50-dimensional and 100-dimensional embeddings were used for \LComments and \News respectively. \footnote{For practical purposes, we take a balanced partition by sampling a split point such that each side has between $(n/3, 2n/3)$ points to ensure a logarithmic query time.}

We also use the \texttt{MNIST}~\cite{lecun1998gradient} dataset for the Hierarchical Clustering Cost experiment discussed earlier in Section~\ref{sec:hccost} and the \texttt{People} dataset from one of the biggest social networking sites for Scalability (Section~\ref{sec:scale}) and Purity (Section~\ref{sec:purity}) experiments. The \texttt{MNIST} dataset is of size $60000 \times 784$, and the

\texttt{People} dataset is of size 40M ($40 \times 10^6$), which is a subsample of the users data on a social networking site.
Each of these members are represented as 50 dimensional vectors. We tie up our experiments in Section~\ref{sec:connecting}. {In Section~\ref{sec:comments}, we discuss the need for accurate predictions in our use-case, i.e., in the \LComments dataset.}

\subsection{Visual Demonstration on Zoo Dataset} 
Using the Zoo~\cite{Dua:2017} dataset, we build a hierarchy using the Hierarchical Clustering proposed in Algorithm~\ref{alg:metahierarchy}. The Zoo dataset consists of 101 animals described with 16 features/attributes and classified into 7 clusters. We use the \EV based splitting rule. The resulting hierarchical clustering is demonstrated in Figure~\ref{fig:zoo}. The labels on the x-axis denote the animal and the color of the label denotes the ground truth cluster. Since ours is an approximation algorithm, we do expect and observe some misclustering. On close observation, we notice that some of the misclustering might be due to the properties of the animals as well. For example, we see that pitviper, seasnake, dolphin and seal are misclassified and are clubbed with aquatic animals.

\subsection{Nearest Neighbor Classification using Hierarchies}
\label{sec:evalNN}
In this task, we use our proposed hierarchical clustering method for the nearest neighbor classification task. We use the tree to find the approximate NN of the query point by traversing the tree to narrow down the candidate set size and doing a brute force search on a small number of points. The pseudocode for the querying and classification mechanism is shown in Algorithm~\ref{alg:metahierarchy}\footnote{Note that though we perform brute force search at the leaf, other efficient methods of classification can equivalently be used inside one cluster/leaf.}.

\begin{table*}[htb!]
\centering
\caption{Performance in terms of Precision(P), Recall(R), F1-Score(F) and Build Time(BT) in seconds, Query Time(QT) in milliseconds per query on 20NG and ALOI datasets (LSH is with $k=8$ and $L = 20$). We highlight the best performance in terms of the F1-Score (F).}
\label{tab:quality_all}
\begin{tabular}{|c|c||ccc|cc||ccc|cc|} 
\hline
                  &                    & \multicolumn{5}{c||}{\textbf{20NG} }               & \multicolumn{5}{c|}{\textbf{ALOI} }                    \\
                  & \textbf{Methods}   & \textbf{P} & \textbf{R} & \textbf{F} & \textbf{BT} & \textbf{QT}  & \textbf{P}  & \textbf{R}  & \textbf{F}  & \textbf{BT} & \textbf{QT}   \\ 
\hline
\multirow{4}{*}{\rotatebox[origin=c]{90}{\textbf{HTNN}}} & \textbf{EV}        & \textbf{0.740}      & \textbf{0.766}       & \textbf{0.752}  &  4.27   &    \textbf{0.59}     & 0.893       & 0.889       & 0.891    & 77.36   & 0.70         \\
                  & \textbf{AEV}       & 0.712      & 0.732      & 0.722      & 7.09 &   0.60      & 0.893       & 0.890       & 0.891    & 70.79   & 0.67          \\
                  
                  & \textbf{RP}        & 0.66      & 0.57      & 0.617    & 1.56  &    0.60     & 0.776       & 0.768       & 0.772   & 37.88  & \textbf{0.55}         \\
                  & \textbf{2-means}   & 0.438     & 0.434      & 0.436    &  19.66 & 0.76        & 0.892       & 0.889       & 0.890    & 183.20   & 0.97          \\ 

\hline\hline
\multirow{4}{*}{\rotatebox[origin=c]{90}{\textbf{LT}}} & \textbf{Average}   & 0.05      & 0.99     & 0.096  & 23.64      &  10.05    & 0.036       & 0.670       & 0.069       & $\sim$3.06k & $\sim$2k      \\
                  & \textbf{Single}    & 0.05      & 0.99      & 0.096      &  15.4  &  10.62     & 0.001       & 0.971       & 0.002   & $\sim$3.06k & $\sim$2k      \\
                  & \textbf{Complete}  & 0.0809      & 0.397      & 0.135 &  23.74      &  7.70     & 0.124       & 0.519       & 0.200       & $\sim$3.06k & $\sim$2k      \\
                  & \textbf{Ward}      & 0.172      & 0.219      & 0.191      & 23.91  &  7.6    & 0.393       & 0.562       & 0.463       & $\sim$3.06k & $\sim$2k      \\ 
\hline\hline
                  & \textbf{SVM}       & 0.613       & 0.608       & 0.610   & 30.16    & 2.06            & 0.843       & 0.830       & 0.836       & -  & -           \\ 
 
\hline\hline

\multirow{4}{*}{\rotatebox[origin=c]{90}{\textbf{ANN}}} & \textbf{LSH}       & 0.559      & 0.552      & 0.555  &  \textbf{0.22}   &    2.14      & \textbf{0.936 }      & \textbf{0.92}        & \textbf{0.934 }   & \textbf{1.02}   & 0.68          \\
             
                & \textbf{kd Tree} & 0.564 & 0.555 & 0.560 & 0.325 & 264 & 0.907 & 0.904 & 0.905 & 4.39 & 421\\
                  & \textbf{kd Tree (w/o}  & \multirow{2}{*}{0.388}      & \multirow{2}{*}{0.37}      & \multirow{2}{*}{0.383}    & \multirow{2}{*}{0.32}  & \multirow{2}{*}{0.89}       &    \multirow{2}{*}{0.825}    & \multirow{2}{*}{0.816}       & \multirow{2}{*}{0.820}   & \multirow{2}{*}{4.37}    & \multirow{2}{*}{7.91}         \\
                  & \textbf{backtrack)} & & & & & & & & & &\\
                  
\hline \hline
& \textbf{PERCH} & 0.047 & 0.0796 & 0.059 & 129.6 & 9.63 & 2.16E-05 & 2.52E-3 & 4.27E-05 & 437.4 & 2.09 \\
\hline
\end{tabular}
\end{table*}
\begin{table*}
\caption{Performance in terms of Precision(P), Recall(R), F1-Score(F) and Build Time (BT) in seconds, Query Time(QT) in milliseconds per query on Covertype and Comments datasets (LSH is with $k=8$ and $L = 20$). We highlight the best performance in terms of the F1-Score (F).}
\label{tab:quality_comments}
\begin{tabular}{|c|c||ccc|cc||ccc|ccc|} 
\hline
                  & \multicolumn{1}{c||}{} & \multicolumn{5}{c|}{\textbf{Covertype}}               & \multicolumn{6}{c|}{\textbf{Comments}}                                                                           \\ 
\hline
                  &                       & \multicolumn{5}{c||}{}                                 & \multicolumn{3}{c|}{ \textbf{Clear vs LQ+Spam}} & \multicolumn{3}{c|}{\textbf{Clear+LQ vs Spam} }  \\
                  & \textbf{Method}      & \textbf{P}  & \textbf{R}  & \textbf{F} & \textbf{BT} & \textbf{QT}  & \textbf{P} & \textbf{R} & \textbf{F}                   & \textbf{P} & \textbf{R} & \textbf{F}                    \\ 
\hline
\multirow{4}{*}{\rotatebox[origin=c]{90}{\textbf{HTNN}}} & \textbf{EV}           & 0.925       & 0.921       & 0.923     & 1330 & 0.38         & \textbf{0.75}       & \textbf{0.84}       & \textbf{0.8}                          & \textbf{0.73}       & \textbf{0.74}       & \textbf{0.74}                          \\
                  & \textbf{AEV}          & 0.925       & 0.922       & 0.923   & 1106   & \textbf{0.37}         & 0.76       & 0.84       & 0.8                          & 0.73       & 0.74       & 0.74                          \\
                  & \textbf{RP}           & 0.904       & 0.897       & 0.901   & 1106   & 0.37         &    0.68       &   0.85        & 0.75                            & 0.61          & 0.55          & 0.58                             \\
                  & \textbf{2-means}      & \textbf{0.930}       & \textbf{0.930}       & \textbf{0.930}   & 1621   & 0.81         & 0.63       & 0.61       & 0.62                         & 0.64       & 0.70       & 0.67                          \\ 
\hline
\hline
                  & \textbf{SVM}          & 0.211       & 0.308       & 0.250   & -   & -            & 0.65          & 0.99          & 0.78                            & 0.735          & 0.431          & 0.54                            \\ 
\hline
\hline
\multirow{4}{*}{\rotatebox[origin=c]{90}{\textbf{ANN}}} & \textbf{LSH}          & 0.581       & 0.739       & 0.650    & \textbf{1.192}  & 44.72        & 0.71       & 0.74       & 0.73                         & 0.62       & 0.63       & 0.62                          \\
                & \textbf{kd Tree} & 0.951 & 0.908 & 0.929 & 14.1527 & 2.57 & 0.69 & 0.88 & 0.77 & 0.706 & 0.560 & 0.62 \\
                  & \textbf{kd Tree (w/o}     & \multirow{2}{*}{0.890}       & \multirow{2}{*}{0.837}       & \multirow{2}{*}{0.863}   &
                  \multirow{2}{*}{14.25}   & \multirow{2}{*}{0.6796}         & \multirow{2}{*}{0.69}          & \multirow{2}{*}{0.79}          & \multirow{2}{*}{0.74}    & \multirow{2}{*}{0.64}          & \multirow{2}{*}{0.48}          & \multirow{2}{*}{0.55}                             \\
                  & \textbf{backtrack)} & & & & & & & & & & & \\

                  \hline \hline
& \textbf{PERCH} & 0.0696 & 0.1428 & 0.093 & 882 & 0.865 & 0.65 & 0.98 & 0.78 & 0.41 & 0.02 & 0.04 \\
\hline
\end{tabular}
\end{table*}

We report our results on the datasets described above for two types of metrics. First, \emph{classification performance}, is classification analysis which includes macro-averaged Precision, Recall, and F1 Scores on the datasets. Second, system performance, compares system time taken for building and querying the models. The \LComments classification task has specific business requirements that require definition of two binary sub-problems for the dataset. First subtask, \Clear v/s (\LQ+\Spam), determines whether a comment should receive unrestricted distribution on the site. Second subtask, (\Clear+\LQ)  v/s \Spam identifies spam comments that should be blocked. 
We compare performance of the following algorithms:
\label{sec:algolist}
\label{sec:evalAlgo}
\begin{enumerate}[labelindent=0pt]
    \item Our proposed methods, Hierarchical Tree Based NN (\HTNN) with splitting rules: \texttt{EV} (exact eigenvector), \ASPD~ (approximate eigenvector), \texttt{RP} (random plane) and \texttt{2-means} (k-means, $k=2$).
    \item Linkage based trees (\LT): Single, Average, Complete and Ward Linkages which minimize minimum, average, maximum and the merged variance between clusters respectively. 
    There is no direct way to query this structure as no distance information at the internal nodes can be used. 
    The tree is cut-off when there are as many clusters as class labels. Each cluster is taken to be as one class. 
    We take pairs of points from the \Test set and then 
    check if the algorithm accurately predicts the pairs of point to be from the same class.
    \item SVM: The support vector machine classifier with the best kernel among (\texttt{polynomial}, \texttt{rbf}, \texttt{sigmoid}).
    \item Approximate Nearest Neighbor Structures (\ANN): The standard kd Tree~\cite{Bentley:1975:MBS:361002.361007} and multiprobe \LSH~\cite{Indyk:1998:ANN:276698.276876}\footnote{\LSH is not a tree based data structure. Tree variants of \LSH exist, but it is unclear how to merge and interpret multiple trees.} used for approximate nearest neighbor searches. For kd-trees, we use standard kd-tree querying mechanism. We also evaluate performance of the kd-tree with querying mechanism as in Algorithm~\ref{alg:metahierarchy}. 
\item PERCH: An Online Hierarchical Clustering algorithm~\cite{kobren2017hierarchical} that maintains a boundary box at every internal node and inserts new points to nodes based on the distances to these boxes.\footnote{PERCH quality is not reported in terms of precision and recall.} We do not evaluate BIRCH~\cite{BIRCH} and BICO~\cite{BICO} as PERCH~\cite{kobren2017hierarchical} is shown to outperform them.
\end{enumerate}

{\subsubsection{Discussion on the Comments Dataset}
\label{sec:comments}
For the \LComments dataset, we treat the problem as two binary classification problems. The first-- \Clear v/s \LQ + \Spam, aims to detect which comments should receive restricted distribution on the site. The second-- \Clear + \LQ v/s \Spam, aims to detect which comments should be completely blocked from the website and should not be visible at all. Let us consider the following cases:
\begin{itemize}
    \item \textbf{\Spam comment marked as \Clear or \LQ}: This case can occur if there is a mis-classification in the second binary classification problem (\Clear + \LQ v/s \Spam). A low quality comment receives restricted access, i.e., only connections of the user who generated the content are exposed to the spam content. A clear comment may be distributed throughout the site. In such a case, many users may be exposed to this spam comment. This impacts the user experience and the site in a negative manner and is undesirable.
    \item \textbf{\LQ comment marked as \Spam}: In this case, a low-quality comment is blocked and not visible to any user. The user who posted the content may feel outraged since the comment is not violating the terms and conditions. This results in a bad user experience, again impacting the site. This case is similar to when a \Clear comment is marked as \Spam.
    \item \textbf{\LQ comment marked as \Clear}: In this case, a low-quality comment is distributed across the site and shown to users who may not be interested in them. For example, if a personal photograph is treated as low quality, people beyond the user's connections may not be interested in it. They may feel that the site is showing irrelevant content and will lose interest in the platform, resulting in negative impact to the site.
\end{itemize}
}
{We treat this problem as a combination of two binary classification problems as it reduces the negative impact to the site and it's reputation. If a \Spam comment is marked as \LQ (the first subproblem), it receives restricted distribution and impacts lesser people as compared to the case when it is given unrestricted distribution. Similarly, marking \LQ as \Clear (the second subproblem) also has lesser negative impact. People may view content as uninteresting but do not have the experience of getting perfectly legitimate content blocked.}

\subsubsection{Results}
In this section, we summarize performance comparisons on the \ALOI, \News datasets (Table~\ref{tab:quality_all}) and \CT, \LComments datasets (Table~\ref{tab:quality_comments}). 

\emph{Classification and Performance Analysis :} Here, we compare performance of the above algorithms in terms of their classification and build/run-time performances. We report performance numbers on \Test set averaged over five runs. From Table~\ref{tab:quality_all}  we see that precision, recall and F1-score of \HTNN is better than \LT on all tasks and performs better than \ANN for three of four datasets. Also, \HTNN~ performs better than \SVM~ on \News, \ALOI and \LComments but for \CT dataset the \SVM algorithm did not converge. \LSH, though has low building time, it does not maintain any hierarchy and is unsuitable for the \News and \LComments datasets. Both of these have an inbuilt hierarchical structure. kd trees demonstrate comparable performance for \ALOI and \CT for the backtracking kd querying mechanism. However, the query times for these are significantly higher as compared to the \HTNN algorithms. To reduce query time for kd tree, we perform querying without backtracking. This experiment shows that the quality degrades and kd trees have highly unbalanced trees in case of sparse data (i.e., \ALOI). As the number of data points increases, memory usage of \LT based data structures becomes unfeasible as observed on the \CT dataset. Query time for \SVM and \LT are also too large to feature in Table~\ref{tab:quality_comments}. {\PERCH performs poorly on all open-source datasets as it builds a highly unbalanced hierarchy where most data points are labelled to be of one class.} All differences in performance are statistically significantly different with five repetitions of all experiments (p<0.05).

\emph{\Clear v/s (\LQ+\Spam)} and \emph{(\Clear+\LQ) v/s \Spam Tasks:} These tasks are binary sub-problems on \LComments dataset, which are important for social network. The (\Clear+\LQ) v/s \Spam task is harder because there is significant overlap between \LQ and \Spam categories. For both tasks, \HTNN outperforms \SVM and \ANN algorithms. For \LComments dataset \LT based data structures are not shown because of time and space consumption which makes them infeasible for our application.
For \PERCH and \SVM, all points are labelled as one class, and thus yield high recall.




\subsection{Anomaly Detection using Hierarchies}
\label{sec:anomaly}

As we mentioned earlier in this paper, a key requirement for our application is to adapt quickly to newly emerging content types. To test the adaptability of our solution, we perform a series of simulations, wherein we selectively drop content from randomly selected subset of classes during training. At test time, we include data from all classes. The Anomaly Detection task is to first decide whether a data point belongs to a known class or to a new \emph{anomaly} class. Post this decision, a label is obtained (in real world, from a human editor) and this data is included in training data and the classifier is updated or retrained. 

In this section, we show that \HTNN with \ASPD based partitioning scheme outperforms in the sense of achieving same or better precision and recall with fewer additional labeled data points.

\begin{table}[htb]
\centering
\caption{Performance in terms of Precision(P), Recall(R), F1-Score(F) for the Anomaly Detection Task on \ALOI~ dataset with \% of points manually labelled. We highlight the best performance in terms of the F1-Score (F).}
\label{tab:anomaly}
\begin{tabular}{|c|cccc|cccc|}
\hline
                           & \multicolumn{4}{c|}{ \textbf{AEV }}                
                           & \multicolumn{4}{c|}{ \textbf{SVM}}                                                     \\
                           
                           & \textbf{\%} & \textbf{P} & \textbf{R} & \textbf{F} & \textbf{\%} & \textbf{P} & \textbf{R} & \textbf{F}  \\
                           \hline \hline
\multirow{3}{*}{\rotatebox[origin=c]{90}{\textbf{ALOI-1}}} & \textbf{51}               & \textbf{0.53}               & \textbf{0.78}            & \textbf{0.63}             & \textbf{68}                      & \textbf{0.42}               & \textbf{0.82}            & \textbf{0.55}              \\ 
                           & 40               & 0.58               & 0.66            & 0.61             & 62                      & 0.42               & 0.76            & 0.54              \\
                           & 30               & 0.62               & 0.53            & 0.57             & 46                       & 0.42               & 0.56            & 0.48              \\
                           \hline \hline
\multirow{3}{*}{\rotatebox[origin=c]{90}{\textbf{ALOI-2}}} & \textbf{51}               & \textbf{0.53}               & \textbf{0.80}            & \textbf{0.64}             & 68                      & 0.42               & 0.84            & 0.56              \\
                           & 40               & 0.58               & 0.66            & 0.62             & \textbf{63}                      & \textbf{0.44}               & \textbf{0.80}            & \textbf{0.56}              \\
                           & 30               & 0.63               & 0.54            & 0.58             & 47                      & 0.45               & 0.61            & 0.52              \\
\hline \hline
\multirow{3}{*}{\rotatebox[origin=c]{90}{\textbf{ALOI-3}}} & \textbf{51}               & \textbf{0.52}               & \textbf{0.78}            & \textbf{0.62}             & 69                      & 0.42               & 0.86            & 0.57              \\
                           & 40               & 0.57               & 0.67            & 62             & \textbf{63}                      & \textbf{0.43}               & \textbf{0.78}            & \textbf{0.55}              \\
                           & 30               & 0.61               & 0.53            & 0.57            & 48                       & 0.43               & 0.60            & 0.50    \\
                           \hline
                         
\end{tabular}
\end{table}
\begin{table}[h]
\centering
\caption{Performance in terms of Precision(P), Recall(R), F1-Score(F) for the Anomaly Detection Task on \News~ dataset with \% of points manually labelled }
\label{tab:newsanamoly}
\begin{tabular}{|c|cccc|cccc|}
\hline
     & \multicolumn{4}{c|}{ \textbf{AEV }} & \multicolumn{4}{c|}{ \textbf{SVM}} \\
            & \textbf{\%}    & \textbf{P}    & \textbf{R}    & \textbf{F}    & \textbf{\%}    & \textbf{P}    & \textbf{R}    & \textbf{F}    \\
\hline \hline
\multirow{3}{*}{\rotatebox[origin=c]{90}{\textbf{20NG-1}}} 
      & 11 & 0.75 & 0.15 & 0.24 & 12 & 0.63 & 0.13 & 0.22 \\
      & 19 & 0.75 & 0.25 & 0.38 & 21 & 0.65 & 0.24 & 0.35 \\
      & \textbf{39} & \textbf{0.71} & \textbf{0.48} & \textbf{0.57} & 42 & 0.65 & 0.48 & 0.55 \\
      \hline \hline
\multirow{3}{*}{\rotatebox[origin=c]{90}{\textbf{20NG-2}}}
      & 10 & 0.70 & 0.12 & 0.21 & 10 & 0.67 & 0.12 & 0.20 \\
      & 18 & 0.70 & 0.22 & 0.34 & 18 & 0.66 & 0.21 & 0.32 \\
      & \textbf{38} & \textbf{0.69} & \textbf{0.47} & \textbf{0.56} & 37 & 0.66 & 0.43 & 0.52 \\
      \hline
\end{tabular}
\end{table}
\begin{table}[htb]
\centering
\caption{Performance in terms of Precision(P), Recall(R), F1-Score(F) for the Anomaly Detection Task on \News~ dataset with \% of points manually labelled, using subclass-superclass relation. We highlight the best performance in terms of the F1-Score (F).}
\label{tab:anomalyHierarchy}
\begin{tabular}{|c|cccc|cccc|}
\hline
      & \multicolumn{4}{c|}{ \textbf{AEV }} & \multicolumn{4}{c|}{ \textbf{SVM}} \\
      & \textbf{\%}   & \textbf{P}    & \textbf{R}     & \textbf{F}     & \textbf{\%}   & \textbf{P}    & \textbf{R}     & \textbf{F}     \\
      \hline \hline
\multirow{3}{*}{\rotatebox[origin=c]{90}{\textbf{20NG-1}}} & 
      11 & 0.93 & 0.66 & 0.77 & 12 & 0.90 & 0.68 & 0.77 \\
& 19 & 0.89 & 0.72 & 0.80 & 21 & 0.84 & 0.72 & 0.78 \\
& \textbf{38} & \textbf{0.81} & \textbf{0.83} & \textbf{0.82} & 42 & 0.76 & 0.82 & 0.79 \\
\hline \hline
\multirow{3}{*}{\rotatebox[origin=c]{90}{\textbf{20NG-2}}} 
& 10 & 0.89 & 0.43 & 0.58 & 10 & 0.86 & 0.37 & 0.52 \\
& 18 & 0.84 & 0.50 & 0.63 & 18 & 0.80 & 0.42 & 0.55 \\
& \textbf{38} & \textbf{0.76} & \textbf{0.68} & \textbf{0.72} & 37 & 0.72 & 0.57 & 0.64 \\
\hline
\end{tabular}
\end{table}
\subsubsection{Data Generation}
We generate 3 datasets from \ALOI and 2 datasets from the \News datasets. For the sake of clarity, we will call the \ALOI generated datasets as ALOI-1, ALOI-2 and ALOI-3. The \News generated datasets will be 20NG-1 and 20NG-2. For the \ALOI datasets, we separate out 5 of 1000 classes to be the anomaly, and from the \News dataset, we separate out 5 of 20 classes. Test set contains of all points from the unseen classes and some points from seen classes. All sampling is done uniformly at random without replacement. ALOI-1, ALOI-2 and ALOI-3 have Train Size $\sim$ 93k, and Test Size $\sim$ 15k, and 20NG-1 and 20NG-2 have Train Size $\sim$ 11k and Test Size $\sim$ 8k.

We also generate 20NG-3 dataset which is used for manual review. In this dataset, we separate 2 classes from 20 classes. This dataset is solely used for gauging the speed of manual reviewing.



\subsubsection{Algorithm Outline}
In this task, we use our hierarchical clustering method for anomaly detection. We build the hierarchy using the Algorithm~\ref{alg:metahierarchy}. We follow a modified version of the classification mechanism used in the Nearest Neighbor Classification task. The two main modifications are -- i) Average pairwise distance for each \textit{known/seen} class is computed and maintained in a table $T$, ii) We calculate the average distance of the nearest neighbors returned through the querying mechanism with the new point (say $d_1$), we look up the predict class in the table $T$ and retrieve the distance $d_2$, iii) For a new point, if $d_1 > d_2$ by a substantial amount (threshold), the point is marked as a possible anomaly and sent for manual review.   

For the baseline, we use \SVM and the prediction probabilities. The intuition behind this is, for any new point, if it belongs to a seen class, the class should be predicted with high probability. If the point is from some unseen or new class, it will be predicted with low probability. Thus, for the baseline, we obtain the prediction probabilities for each class. If the highest probability among all probabilities is fairly low (based on some threshold), the point is marked as a possible anomaly.

All possible anomalies have to be manually labelled. This labelling is associated with employing human labelers and in turn, incurs high cost. Humans also take much more time and thus, the company would need to employ many manual labelers to gain high accuracy. Thus, we vary the thresholds to obtain \% which denotes the percentage of points that are sent for manual labelling. This \% is an indicator of the manual labelling cost or budget for the company. Higher the \%, higher will be the cost to the company.

\subsubsection{Results}
In Table~\ref{tab:anomaly}, we present results on ALOI-1, ALOI-2 and ALOI-3 datasets. As mentioned, \% denotes the percentage of new points sent for manual labelling. This is an indicator of the cost for manual labelling. From Table~\ref{tab:anomaly}, we observe that with lower \% or lower additional labeling cost, we obtain higher F1-Score using \ASPD~ based algorithm as compared to the \SVM based algorithm. \ASPD~ based algorithm, with 30\% of manual labelling, performs similar to \SVM~ based algorithm with 68\% of manual labelling. This trend is observed for all \ALOI generated datasets, i.e., ALOI-1, ALOI-2 and ALOI-3. This indicates that using \ASPD~ based algorithm should perform well in practice.

As we mentioned earlier, for newly discovered content classes, our application allows for applying the treatment of the parent class in the hierarchy. We perform another experiment that exploits the inherent \texttt{hierarchy} in the \News~ dataset. Here, we first present the results on the original algorithm in Table~\ref{tab:newsanamoly}. We then modify the algorithm to consider the hierarchy as follows-- if the superclass of an \textit{unseen}/\textit{anomaly} class is predicted correctly, we treat it as a correct classification. The reason for this is, in the social networking site application for example if a new type of \Spam~ arrives and if it is correctly marked as \Spam~, the site behavior is maintained as per requirement. We observe that the F1-Scores are better in case of the \ASPD~ based algorithm as compared to the \SVM~ based algorithm (see Table~\ref{tab:anomalyHierarchy}). In this experiment, the \ASPD~ based algorithm with 20\% of labelling performs similar to \SVM~ based algorithm with $~$40\% labelling. \ASPD~ based algorithm thus performs well for data with high overlap and an inherent hierarchical structure like the 20NG data.

\subsubsection{Manual review of Anomalous Content}
\label{sec:manual}
For editors, having a hierarchy makes it easier to label new classes. In order to demonstrate this, we performed an experiment with experienced labelers at the social networking site. We took \texttt{20 Newsgroups} dataset and separated 2 classes, namely \texttt{sci.electronics} and \texttt{politics-misc} (i.e., 20NG-3), sampled some points from each of the 20 classes (18 known classes and 2 new classes). In one experiment, we gave labeler set 1 a list of all known classes and the option to put it in a new class. For labeler set 2, we give them the hierarchy of classes as well. We measured the time taken to label the points. On average, the first set of labelers, without the hierarchy, took about \textbf{141 seconds per document}. The second set, with the hierarchy, took only \textbf{100 seconds per document}.

\subsection{Scalability}
\label{sec:scale}

The \EV, \ASPD, \texttt{RP} and \texttt{2-means} algorithms are highly scalable and can leverage \emph{Apache Spark} and the \emph{mllib} library. The core operations of the proposed hierarchical clustering algorithms -- {\em i) dot products, ii) matrix multiplication, iii) eigenvector estimation and, iv) k-means algorithm.} The first two operations and the k-means algorithm are directly supported and parallelizable through the \emph{Spark mllib} library. Eigenvector estimation can be done either via matrix multiplications or through using Singular Value Decomposition from the library. 

We implemented the proposed hierarchical clustering algorithm with \ASPD~ and \texttt{RP} splitting rules in Java Spark and use it to cluster fairly big data. We sub-sample \textbf{40M} = $40 \times 10^6$ records from the People dataset of LinkedIn and use it as the training data. The build time using the \ASPD~ algorithm was $\sim$ 3 hours, where each record was a $50$ dimensional point and the tree was built such that there were 16384 leaf clusters. Using the \texttt{RP} based splitting rule, the build time was $\sim$ 2 hour. For both these evaluations, we used 30 executors, each having 1 core and 10GB of memory. We also build a hierarchy of 32768 clusters and use the two hierarchies to perform querying on a dataset of \textbf{264M} records. We vary the number of executors and report the results in terms of queries per second (QPS) in Table~\ref{tab:scalability}.
\begin{table}[h]
    \caption{Queries per second (QPS) with varying number of executors for the People dataset.}
    \label{tab:scalability}
    \centering
    \begin{tabular}{|c||c|c|c|c|}
        \hline
        & \multicolumn{4}{c|}{\textbf{Number of executors}}\\ 
        & \textbf{10} & \textbf{20} & \textbf{50} & \textbf{100} \\
        \hline
        \hline
        \textbf{16384} & 118918 & 209524 & 293333 & 314285 \\
        \textbf{32768} & 115789 & 191304 & 275000 & 307692 \\
        \hline
    \end{tabular}
\end{table}

\subsection{Purity}
\label{sec:purity}
In this particular task, we use \ASPD~ based strategy to build a deep hierarchy with almost completely pure leaf clusters. We define purity-- 
$$
    \mathcal{P} = \frac{1}{|\mathcal{C}|} \sum_{\mathcal{C}_i \in \mathcal{C}} \frac{1}{|\mathcal{C}_i|}\max_{l \in labels} count(l, \mathcal{C}_i)
$$
where $\mathcal{C}$ is the set of all clusters (leaves in case of hierarchy), $l$ is the set of all labels, and, 
$count(l, \mathcal{C}_i)$ is the number of points in cluster $\mathcal{C}_i$ with label $l$. Intuitively, purity measures the average cleanness of clustering. The higher the purity, the more the cleanness.

We evaluate the purity v/s query time for the \News, \LComments, \ALOI datasets. We also perform the experiment on the \texttt{People} dataset of the social networking site. For the \texttt{People} dataset, we use the \textit{member industry} tag as a label and evaluate the purity v/s query time. The query times are averaged for 8000 points for the \News dataset, 20000 for \ALOI and \LComments datasets and 1M for \texttt{People} dataset. The query time is reported as the average time for one query.

As a baseline, we perform flat, \texttt{k-means} clustering. The query time for this task is the \emph{cluster assignment time}, denoted by \textbf{T}. We obtain two types of results using \texttt{k-means} -- i) when the query time is similar to the query time in the hierarchy, and, ii) when the purity is similar to that of the leaf clusters in the hierarchy. To perform this evaluation, we first build a hierarchy till the leaves are almost completely pure, and use this hierarchy to obtain query times. We also obtain the number of leaves in this hierarchy (say $t$). Since we want to compare the purity and query time for \texttt{k-means}, we perform \texttt{k-means} with i) number of clusters, $k = \log_2(t)$ -- this gives us a query time which is similar to that of the hierarchy, and ii) $k = t$ -- this gives us a purity similar to that of the hierarchy. For \texttt{k-means}, we use the k-means method in scikit-learn. For the \texttt{People} dataset, we built the hierarchy with $~16000$ clusters and performed \texttt{k-means} with $k = 16000$ clusters.

The results are tabulated in Table~\ref{tab:purity}. Here, we denote \texttt{k-means} results by k1 and k2, where k1 is the case where we obtain similar cluster lookup time, or, query time as the \ASPD~ based approach, and, k2 is when we obtain high purity clusters (of similar purity as \ASPD~ based approach) using k-means algorithm. \ASPD~ based clustering consistently outperforms the \texttt{k-means} based clustering.

A low purity indicates that we can make errors in classification of a new point, which is quite undesirable for our use case. If a \Spam type is not blocked, it can have various effects leading to user dissatisfaction. Also, if a \Clear type is blocked, the user may be outraged and this again leads to user dissatisfaction. Having high query time also impacts the performance as a comment would not be available for a noticeably long period of time. Thus, our particular use case demands high purity and low cluster assignment times. In such a scenario, the hierarchical clustering algorithm clearly outperforms all other reasonable baselines.

\begin{table}[htb]
    \centering
    \caption{Purity (P) v/s Query Time per point (T) in milliseconds for various datasets}
    \label{tab:purity}
    \begin{tabular}{|c|cc|cc|cc|cc|}
        \hline
        & \multicolumn{2}{c|}{\textbf{20NG}} & \multicolumn{2}{c|}{\textbf{Comments}} & \multicolumn{2}{c|}{\textbf{ALOI}} & \multicolumn{2}{c|}{\textbf{People}} \\
        & \textbf{P} & \textbf{T} & \textbf{P} & \textbf{T} & \textbf{P} & \textbf{T} & \textbf{P} & \textbf{T}\\
        \hline
        \textbf{EV} & 0.99 & 0.14 & 0.99 & 0.07 & 0.99 & 0.117 &  0.26 & 0.0007\\
        \hline
        \textbf{k1} & 0.69 & 0.17 & 0.72 & 0.10 & 0.34 & 0.114 &  - & -\\
        \hline
        \textbf{k2} & 0.91 & 0.98 & 0.9 & 2.34 & 0.96 & 1.763 &  0.25 & 0.002\\
        \hline
    \end{tabular}
\end{table}

\subsection{Hierarchical Clustering Cost}
{Exact calculation of hierarchical clustering objective requires $O(n^2 \log n)$ time and $O(n^2)$ space, impractical for large data. 
However, we evaluate the cost for \texttt{MNIST} (Train set = 60k, Dimensions = 784), \texttt{20NG} (Train set = 11.8k, Dimensions = 100) and \texttt{ALOI} (Subsampled train set = 60k, Dimensions = 128) datasets for various algorithms. Since the costs are of the order of $10^{12}-10^{14}$, we report the relative cost in Table~\ref{tab:cost}. 
For all three datasets, \EV and \AEV obtains similar cost, \RP and \texttt{2-means} obtain relatively higher costs. 
All linkage based algorithms obtain significantly worse cost as compared to the \EV and \AEV based algorithms.\footnote{Note that for the cost experiment, since we are dealing with small datasets, we perform the exact split instead of the (1/3, 2/3) split used in other experiments.}}

\begin{table}[h]
\centering
\caption{Relative Hierarchical Clustering Cost for MNIST, 20NG and ALOI datasets.}
\label{tab:cost}
\begin{tabular}{|c||c|c|c|c|c|c|c|c|}
\hline
& \multicolumn{4}{c|}{\textbf{HTNN}} & \multicolumn{4}{c|}{\textbf{LT}} \\
               & {\textbf{EV}} & {\textbf{AEV}} & {\textbf{RP}} & {\textbf{2-means}} & {\textbf{Single}} & {\textbf{Average}} & {\textbf{Complete}} & {\textbf{Ward's}} \\
\hline
\textbf{MNIST} & 1                               & 1.01                             & 1.05                            & 1.1                                  & 1.17                                & 1.13                                 & 1.12                                  & 1.13                                \\
\textbf{20NG}  & 1.01                            & 1                                & 1.01                            & 1.05                                 & 1.14                                & 1.14                                 & 1.14                                  & 1.14                                \\
\textbf{ALOI}  & 1                               & 1                                & 1.03                            & 1.09                                 & 1.16                                & 1.1                                  & 1.11                                  & 1.1  \\               
\hline
\end{tabular}
\end{table}

\subsection{Discussions}
\label{sec:connecting}
{All tasks were defined keeping in mind specific business requirements and with the intention of measuring the quality of the clustering. The requirements state that we need a highly unsupervised, near real-time algorithm to detect newer types of content. 
Our first task builds a hierarchy on the Zoo dataset and is used solely for visualization purposes. This task demonstrates the quality of our hierarchy visually.
Our second task in Section~\ref{sec:evalNN} elaborates on how to use the hierarchy to classify the classes of incoming, user generated content. Considering the trade-off between query time and the F1-score, our proposed method out performs other baselines for both text based datasets and for Covertype dataset. Our third task in Section~\ref{sec:anomaly} makes use of the classification task and outlines an algorithm which uses the hierarchy to identify new content types. We also demonstrate how the subtype-supertype relationship can be used to provide quick responses for the incoming data. In Section~\ref{sec:manual}, we also estimate the benefits of using a hierarchy for manual labelling of anomalous content. Since manual labelling of new types requires human intervention and effort, it is ideal to borrow the response actions from the supertype over sending it for manual labelling. The fourth task, Section~\ref{sec:scale}, demonstrates the ability to scale to multi-million sized data by exploiting parallelizability of various computations. The fifth and final tasks in Section~\ref{sec:purity}, explores the `cleanliness` of the leaves of the hierarchy and the query time obtained using the hierarchy as compared to a flat clustering using \texttt{k-means}. Since our use case targets classifying content while identifying anomalous and adversarial content at scale while maintaining low query time and high purity clusters, employing a hierarchical clustering is the only natural approach.}

\section{Conclusion}

In this paper, showed that several partitioning schemes are suitable for creating hierarchical clustering at scale for the purpose of content classification. Our techniques are practical and have approximation guarantees for a specific objective. We demonstrated that our methods scale well at training as well as at runtime, adapt efficiently to novel content classes, and have favorable hierarchical cost characteristics. Our method is thus suitable to the social networking sites' application requirements as mentioned at the beginning of this paper. We also showed that our method performs well on several open source datasets and other performant algorithms.  

Adapting our method to an online, incremental, learning setting remains as future work.

\section*{Acknowledgement}
We are grateful to the anonymous reviewers for their helpful feedback. This project has received funding from the Engineering and Physical Sciences Research Council, UK (EPSRC) under Grant Ref: EP/S03353X/1. Anirban  acknowledges the kind support of the N. Rama Rao Chair Professorship at IIT Gandhinagar, the Google India AI/ML award (2020), Google Faculty Award (2015), and CISCO University Research Grant (2016).

\bibliographystyle{unsrt}  
\bibliography{references}

\end{document}